\documentclass[runningheads,envcountsame]{llncs}

\newif\ifijcarversion 
\ijcarversiontrue

\usepackage[T1]{fontenc}
\usepackage{mathtools,stmaryrd,amssymb,amsfonts,nicefrac,bm}
\usepackage{subcaption}
\usepackage{hyperref}
\usepackage[capitalise,noabbrev]{cleveref}
\usepackage{cite}
\usepackage{enumitem}
\usepackage{subfiles}
\usepackage{proof}

\usepackage{tikzit}

\tikzstyle{gate}=[fill=white, draw=black, shape=rectangle, minimum height=0.43cm, minimum width=0.43cm, inner sep=0.1em]
\tikzstyle{control}=[fill=black, draw=black, shape=circle, scale=0.38]
\tikzstyle{not}=[shape=circle, path picture={ 
\draw[black](path picture bounding box.north) -- (path picture bounding box.south) (path picture bounding box.west) -- (path picture bounding box.east);
}, draw=black, scale=.8]
\tikzstyle{wcontrol}=[fill=white, draw=black, shape=circle, scale=0.38]
\tikzstyle{bwcontrol}=[draw=black, shape=circle, scale=0.38, path picture={
        \fill[white] (path picture bounding box.center) circle(0.5);
        \fill[black] (path picture bounding box.center) [radius=0.5] -- ++(0:0.5) arc[start angle=0, end angle=180, radius=0.5] -- cycle;
    }, rotate=45]
\tikzstyle{empty}=[fill=white, draw=black, shape=rectangle, inner sep=0.4em, emptyborder]
\tikzstyle{globalphase}=[fill=white, draw=black, inner sep=0.15em, shape=rounded rectangle, minimum height=0.4cm]
\tikzstyle{ancilla}=[fill=black, draw=black, shape=rectangle, minimum width=0.01cm, minimum height=0.25cm, inner sep=0.01em]
\tikzstyle{ground}=[fill=white, path picture={\draw[black](-1.5mm,0)--(-0.6mm,0);\draw[black,thick](-0.6mm,-1.75mm)--(-0.6mm,1.75mm) (0mm,-0.9mm)--(0mm,0.9mm) (0.6mm,-0.5mm)--(0.6mm,0.5mm);}, minimum width=0.1mm, draw=none, outer sep=0pt]
\tikzstyle{gate22}=[fill={rgb,255: red,220; green,220; blue,220}, draw=black, shape=rectangle, minimum height=.68cm, minimum width=0.6cm]
\tikzstyle{void}=[shape=rectangle, minimum height=0.5cm]
\tikzstyle{gate44}=[fill={rgb,255: red,220; green,220; blue,220}, draw=black, shape=rectangle, minimum height=1.43cm, minimum width=0.5cm]
\tikzstyle{divider}=[fill={rgb,255: red,220; green,220; blue,220}, draw=black, shape border rotate=90, regular polygon, regular polygon sides=3, inner sep=1.5pt, rounded corners=0.5mm]
\tikzstyle{gatherer}=[fill={rgb,255: red,220; green,220; blue,220}, draw=black, shape border rotate=-90, regular polygon, regular polygon sides=3, inner sep=1.5pt, rounded corners=0.5mm]
\tikzstyle{hyperedge}=[fill=white, draw=black, shape=rectangle, rounded corners=0.1cm, minimum height=.6cm, minimum width=.6cm]
\tikzstyle{square}=[fill=white, draw=black, shape=rectangle, minimum height=0.20cm, minimum width=0.20cm, inner sep=0.1em, thick]
\tikzstyle{gphase}=[rounded rectangle, rounded rectangle arc length=120, fill={zx_grey}, inner sep=2pt, font={\tiny\boldmath}, label distance=1mm, fill opacity=.8, text opacity=1, tikzit category=ZX]
\tikzstyle{customcontrol}=[fill=white, draw=black, inner sep=0.1em, shape=rounded rectangle, minimum height=0.2cm]
\tikzstyle{whiteancilla}=[fill=white, draw=black, shape border rotate=-90, regular polygon, regular polygon sides=3, inner sep=1.5pt, rounded corners=0.2mm]
\tikzstyle{blackancilla}=[fill=black, draw=black, shape border rotate=-90, regular polygon, regular polygon sides=3, inner sep=1.5pt, rounded corners=0.2mm]
\tikzstyle{greyancilla}=[fill={rgb,255: red,150; green,150; blue,150}, draw=black, shape border rotate=-90, regular polygon, regular polygon sides=3, inner sep=1.5pt, rounded corners=0.2mm]
\tikzstyle{whiteancillaterm}=[fill=white, draw=black, shape border rotate=90, regular polygon, regular polygon sides=3, inner sep=1.5pt, rounded corners=0.2mm]
\tikzstyle{blackancillaterm}=[fill=black, draw=black, shape border rotate=90, regular polygon, regular polygon sides=3, inner sep=1.5pt, rounded corners=0.2mm]

\tikzstyle{emptyborder}=[-, dash pattern=on 0.16em off 0.16em on 0.16em off 0.16em on 0.16em off 0em]
\tikzstyle{etc}=[-, draw=black, densely dashed, thick]
\tikzstyle{greyetc}=[-, draw={rgb,255: red,161; green,161; blue,161}, densely dashed, thick]
\tikzstyle{dots}=[-, dotted, draw=black, thick]
\tikzstyle{big}=[-, thick]
\tikzstyle{register}=[-, double]
\tikzstyle{grey}=[-, draw={rgb,255: red,161; green,161; blue,161}]
\tikzstyle{border}=[-, fill=white]

\input{qc.tikzdefs}
\newcommand{\tf}[1]{\scalebox{0.77}{\input{#1.tikz}}}

\newcommand{\smalltf}[1]{\scalebox{0.72}{\input{#1.tikz}}}
\newcommand{\tinytf}[1]{\scalebox{0.64}{\input{#1.tikz}}}

\usepackage{xcolor}

\hypersetup{
  colorlinks,
  linkcolor={blue!80!black},
  citecolor={blue!80!black},
  urlcolor={blue!80!black}
}

\allowdisplaybreaks

\newcommand{\labelcrefor}[1]{%
  \begingroup
  \renewcommand{\crefpairconjunction}{ or }%
  \renewcommand{\crefmiddleconjunction}{, }%
  \renewcommand{\creflastconjunction}{ or }%
  \labelcref{#1}%
  \endgroup
}

\newcommand{\N}{\mathbb{N}}

\newcommand{\R}{\mathbb{R}}

\newcommand{\natset}[1]{\{1,\dots,#1\}}

\newcommand{\defeq}{\coloneqq}

\newcommand{\semicolon}{\fatsemi}
\newcommand{\interp}[1]{\llbracket #1 \rrbracket}

\newcommand{\cat}[1]{\textbf{\textup{#1}}}

\newcommand{\rewrite}[1]{\longrightarrow_{#1}}
\newcommand{\mrewrite}[1]{\longrightarrow_{#1}^*}

\newcommand{\bothmrewrite}[1]{\longleftrightarrow_{#1}^*}
\newcommand{\notrewrite}[1]{\,\not\!\!\longrightarrow_{#1}}
\newcommand{\first}{\raisebox{1pt}{\scalebox{0.6}{$\#$}}}
\newcommand{\emptylist}{[\,]}
\newcommand{\conc}{\mathbin{\raisebox{-1pt}{\textup{@}}}}
\newcommand{\smallconc}{\,\mathbin{\raisebox{-0.5pt}{\scalebox{0.7}{$\textup{@}$}}}\,}
\newcommand{\listplus}{+_{\!\!\cdot}}
\newcommand{\listminus}{-_{\!\!\cdot}}

\newcommand{\terms}[1]{\langle#1\rangle}
\newcommand{\contexts}[3]{\langle#1\rangle_{#2,#3}}
\newcommand{\ppcontexts}[3]{\langle#1\rangle^{\smash{\textup{pp}}}_{\smash{#2,#3}}}
\newcommand{\pp}[1]{\langle#1\rangle^{\textup{pp}}}
\newcommand{\nf}[1]{\langle#1\rangle^{\textup{nf}}}
\newcommand{\cf}[1]{\langle#1\rangle^{\textup{cf}}}
\newcommand{\layer}[3]{L^{\smash{#1}}_{\smash{#2}}(#3)}
\newcommand{\layers}[3]{\vec{L}^{\smash{#1}}_{\smash{#2}}(#3)}
\newcommand{\clayers}[3]{\vec{C}^{\smash{#1}}_{\smash{#2}}(#3)}
\newcommand{\ceta}{\textsf{C\kern-0.15exe\kern-0.45exT\kern-0.45exA}}

\title{Towards Term-based Verification of Diagrammatic Equivalence}
\titlerunning{Towards Term-based Verification of Diagrammatic Equivalence}
\author{Julie Cailler\inst{1} \and 
Noé Delorme\inst{1} \and 
Simon Perdrix\inst{1} \and 
Sophie Tourret\inst{1,2}}

\authorrunning{J. Cailler, N. Delorme, S. Perdrix, S. Tourret}

\institute{Université de Lorraine, CNRS, INRIA, LORIA, Nancy, France\\ \and Max Planck Institute for Informatics, Saarbrücken, Germany
\email{firstname.lastname@inria.fr}}

\begin{document}

\maketitle

\begin{abstract}
  A string diagram is a two-dimensional graphical representation that can be described as a one-dimensional term generated from a set of primitives using sequential and parallel compositions. Since different syntactic terms may represent the same diagram, this syntax is quotiented by a collection of coherence equations expressing equivalence up to deformation. This work lays foundations for automated reasoning about diagrammatic equivalence, motivated primarily by the verification of quantum circuit equivalences. We consider two classes of diagrams, for which we introduce normalizing term rewriting systems that equate diagrammatically equivalent terms. In both cases, we prove termination and confluence with the help of the proof assistant Isabelle/HOL. 
\end{abstract}


\section{Introduction}

Diagrammatic reasoning has proved to be a powerful and expressive framework across a wide range of domains, including quantum computing \cite{CD08,jeandel2018complete,zw,backens2018zh,AC04,qccompleteness}, linear algebra \cite{sobocinski2015graphical,bonchi2014interacting}, boolean circuits \cite{categorytof}, and natural language processing \cite{coecke2021mathematics,zeng2016quantum}. By representing computations as compositional diagrams in two dimensions, these formalisms offer intuitive and structural insights that are difficult to witness otherwise. As diagrammatic methods continue to improve, the need for automated reasoning techniques expands. In particular, problems such as equivalence checking, verification, certification, and diagrammatic manipulation require algorithmic foundations in order to scale.

An important motivation for our work arises from quantum computing, a field in which diagrammatic reasoning has been extensively developed \cite{qccompleteness,qcextensions,qcminimality,controlledprop}. In particular, quantum circuits admit natural diagrammatic representations, and reasoning about their equivalence is essential for tasks such as optimization, compilation, and verification. Our work is motivated by the goal of obtaining semi-automated verification techniques for quantum circuit equivalences. In this paper, we establish foundational results, in the form of normal forms and verified rewriting systems, that constitute essential building blocks towards a full-fledged proof certification pipeline for quantum circuits.

One way to represent diagrammatic reasoning problems suitably for automation is to adopt a syntactic term-based representation of diagrams. By encoding diagrams as terms equipped with operations for sequential and parallel compositions, diagrammatic deformations can be translated into equations between terms. This enables the use of term rewriting to prove diagram equivalences, including standard techniques for proving termination and confluence, for which there exist automated tools \cite{aprove,ceta,maud,egglog}. As a result, the term-based representation provides a natural framework for certifying diagrammatic equivalence in a rigorous and machine-checkable way.

The natural mathematical framework to encode diagrams as terms is monoidal categories, and more precisely PROs and PROPs \cite{selingerbible}. Both structures are similar, with one distinction: a PROP additionally allows swapping wires. Formally, the terms are freely generated from a set of generators---corresponding to black box processes with a certain number of input and output wires---and can be composed in two ways: in sequence with the operator $\semicolon$ and in parallel with the operator $\otimes$. Depending on the order in which the generators are syntactically composed, multiple terms may represent the same diagrammatic evolution. For instance, $(u_1\semicolon u_2)\otimes(v_1\semicolon v_2)$ and $(u_1\otimes v_1)\semicolon(u_2\otimes v_2)$ can both be graphically depicted as follows.
\begin{equation*}
    \tf{interchange}
\end{equation*}

Hence, the syntax of terms is quotiented by a collection of \emph{coherence equations} capturing the intuitive notion of diagram equivalence. These equations can be turned into rewrite rules, allowing the rewrite of equivalent terms.

\paragraph{Contributions.}
We focus on two diagrammatic equivalence problems. In the first one, we consider arbitrary state-and-effect-free PROs, meaning that all generators have at least one input wire and one output wire. We define a normal form together with an associated normalizing rewriting system, and prove termination and confluence. Then, any pair of diagrammatically equivalent terms can be rewritten one into the other using the normalization procedure. In the second problem, we consider the PRO of permutations, which is precisely the PROP with no generators. Building upon our previous normal form, we define a canonical form and a canonizing rewriting system that we also prove to be terminating and confluent. This ensures the existence and uniqueness of the canonical form for any permutation. In both cases, the key steps of the termination and confluence proofs are verified using the proof assistant Isabelle/HOL. 

\paragraph{Related Work.}
Diagrammatic equivalence has also been studied with a fundamentally different approach: instead of working with a syntactic term representation, diagrams are encoded as hypergraphs and rewriting is studied at the level of graph structures \cite{stringdiagramtheory}. Then, diagrammatic equivalence is witnessed by hypergraph isomorphism. In contrast, our work remains entirely in the term-based setting and directly equips terms with rewriting systems yielding explicit normal forms that we plan to use for certification. 
Related questions have also been investigated in the context of the ZX-calculus, notably with the tool VyZX \cite{vyzx}. While sharing the same overall goal of certifying diagrammatic equivalences, their approach focuses on filling the gap between the use of inductive datatypes in proof assistants and the use of graph representation in ZX-calculus. Diagrammatic equivalence in the PRO of permutations has already been addressed, but from a diagrammatic perspective, where equivalence is handled implicitly via graphical transformation \cite{Lafont2003}. Although the resulting normal form is closely related in shape to ours, this work remains purely graphical and does not account for the syntactic order of sequential and parallel compositions.

\paragraph{Outline of the Paper.}
In \cref{sec:background}, we fix the notations and give term-based definitions of PROs and PROPs. In \cref{sec:pro}, we introduce the normal form for an arbitrary state-and-effect-free PRO and prove the termination and confluence of its underlying term rewriting system. In \cref{sec:perm}, we introduce the canonical form for the PRO of permutations and prove the termination and confluence of its underlying term rewriting system. In \cref{sec:conclusion}, we provide our concluding remarks. The Isabelle companion files are available on Zenodo \cite{delorme_2026_20294284}.


\section{Background}\label{sec:background}

The standard mathematical framework for diagrammatic reasoning is the one of \emph{monoidal categories}. In particular, circuit-based graphical languages---such as quantum circuits---are defined as a particular kind of monoidal categories, namely PROs or PROPs (which stands for \emph{PROduct and Permutation categories}). A PRO is a strict monoidal category generated by a single object, and a PROP is additionally symmetric (meaning that there are swap primitives allowing wires to be exchanged). For non-category theorists, we provide self-contained definitions of PROs and PROPs in this section. Roughly speaking, these structures are generated from a set of generators---corresponding to basic gates---via sequential and parallel composition. Graphically, this amounts to representing the generators as boxes and connecting them with wires.

\subsection{Terms and their Diagrammatic Representation}

Let $\Sigma$ be a set of generators equipped with two functions $\textup{dom},\textup{cod}:\Sigma\to\N$. We inductively define the set of terms $\terms{\Sigma}$ induced by $\Sigma$, and extend $\textup{dom},\textup{cod}:\terms{\Sigma}\to\N$ as follows, using the notation $t:n\to m\in \terms{\Sigma}$ (or simply $t:n\to m$ if $\Sigma$ is clear from the context) when $t\in \terms{\Sigma}$,  $\textup{dom}(t)=n$ and $\textup{cod}(t)=m$.
\begin{gather*}
    \infer{g:\textup{dom}(g)\to \textup{cod}(g) \in\terms{\Sigma}}{g\in\Sigma} \hspace{3em}
    \infer{\textup{id}_n:n\to n\in\terms{\Sigma}}{n\in\N} \\[0.2em]
    \infer{(u\semicolon v):n\to m\in\terms{\Sigma}}{u:n\to p \in\terms{\Sigma} & v:p\to m\in\terms{\Sigma}} \hspace{3em}
    \infer{(u\otimes v):n+p\to m+q\in\terms{\Sigma}}{u:n\to m\in\terms{\Sigma}&v:p\to q\in\terms{\Sigma}}
\end{gather*}

The \textit{domain} and \textit{codomain} of a term $t\in\terms{\Sigma}$ are given by $\textup{dom}(t)$ and $\textup{cod}(t)$, representing its number of inputs and outputs, respectively. We define the \emph{defect} of $t$ by $\textup{def}(t)\defeq \textup{dom}(t)-\textup{cod}(t)$. A term $t$ is called a \emph{state} if $\textup{dom}(t)=0$, and an \emph{effect} if $\textup{cod}(t)=0$. Moreover, for each $t\in\terms{\Sigma}$, let $\textup{gen}(t)\subseteq\Sigma$ denote the set of generators appearing in $t$. This set is defined inductively by $\textup{gen}(\textup{id}_n)=\varnothing$, $\textup{gen}(g)=\{g\}$ for $g\in\Sigma$, and $\textup{gen}(u\semicolon v)\defeq \textup{gen}(u\otimes v)\defeq \textup{gen}(u)\cup\textup{gen}(v)$.

A common set of generators is the set of swaps $\Sigma_0\defeq\{\sigma_{n,m}:n,m\in\N\}$ with $\textup{dom}(\sigma_{n,m})=\textup{cod}(\sigma_{n,m})= n+m$. Intuitively, $\sigma_{n,m}$ represents a diagram that exchanges two registers consisting of $n$ and $m$ wires, respectively. Terms admit a natural graphical representation as string diagrams. A term $t:n\to m$ is depicted as a box with $n$ input wires and $m$ output wires; the identity $\textup{id}_n$ is depicted as $n$ parallel wires; and the swap $\sigma_{n,m}$ is depicted as a swap of the $n$ and $m$ wires. Here are some examples.
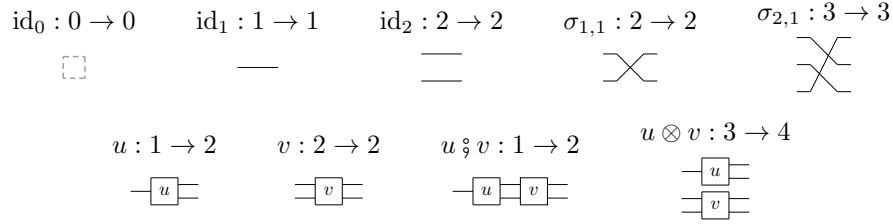
\begin{gather*}
    \begin{array}{c}
        \textup{id}_0:0\to 0 \\[0.2em] \smalltf{empty}
    \end{array} \hspace{2em}
    \begin{array}{c}
        \textup{id}_1:1\to 1 \\[0.2em] \smalltf{id1}
    \end{array} \hspace{2em}
    \begin{array}{c}
        \textup{id}_2:2\to 2 \\[0.2em] \smalltf{id2}
    \end{array} \hspace{2em}
    \begin{array}{c}
        \sigma_{1,1}:2\to 2 \\[0.2em] \smalltf{sigma11}
    \end{array} \hspace{2em}
    \begin{array}{c}
        \sigma_{2,1}:3\to 3 \\[0.2em] \smalltf{sigma21}
    \end{array} \\[0.4em]
    \begin{array}{c}
        u:1\to 2 \\[0.2em] \smalltf{u12}
    \end{array} \hspace{2em}
    \begin{array}{c}
        v:2\to 2 \\[0.2em] \smalltf{v22}
    \end{array} \hspace{2em}
    \begin{array}{c}
        u\semicolon v:1\to 2 \\[0.2em] \smalltf{sequ12v22}
    \end{array} \hspace{2em}
    \begin{array}{c}
        u\otimes v:3\to 4 \\[0.2em] \smalltf{paru12v22}
    \end{array}
\end{gather*}

Moreover, given $r,s\in\N$, the set of contextualized terms $\contexts{\Sigma}{r}{s}$ is defined as the terms containing a single occurrence of an additional generator $[\cdot]:r\to s$ called a \emph{hole}. Then, given a contextualized term $K[\cdot]:n\to m\in\contexts{\Sigma}{r}{s}$ and a term $t:r\to s\in\terms{\Sigma}$, the substitution $K[t]:n\to m\in\terms{\Sigma}$ is the term where the unique hole $[\cdot]$ in $K[\cdot]$ has been replaced by $t$. 

\subsection{Diagrammatic Equivalence}

Up to here, terms are purely syntactic, and two terms are only equal if they are constructed exactly in the same way. However, several terms may be represented by the same diagram. Thus, we need the equivalence equations known as the \emph{coherence equations}, depicted in \cref{fig:coherencelaws}, to precisely capture the intuitive notion of diagram deformation, which yields our definition of PROs and PROPs.

\begin{figure}[t]
    \fbox{\begin{minipage}{0.982\linewidth}\begin{center}
        \vspace{-1em}
        \hspace{-0.8em}
        \begin{subfigure}{0.48\textwidth}
            \begin{equation}\label{eq:seqidentity}
                \begin{array}{ccccc}
                    \textup{id}_n\semicolon t&=&t&=&t\semicolon\textup{id}_m\\[.4em]
                    \tinytf{seqidentity_left}&=&\tinytf{seqidentity_mid}&=&\tinytf{seqidentity_right}
                \end{array}
            \end{equation}
        \end{subfigure}

        \vspace{-0.5em}
        \begin{subfigure}{0.44\textwidth}
            \begin{equation}\label{eq:seqassociativity}
                \begin{array}{ccc}
                    (t\semicolon u)\semicolon v&=&t\semicolon(u\semicolon v)\\[.4em]
                    \tinytf{seqassociativity_left}&=&\tinytf{seqassociativity_right}
                \end{array}
            \end{equation}
        \end{subfigure}\hspace{2em}
        \begin{subfigure}{0.38\textwidth}
            \begin{equation}\label{eq:parassociativity}
                \begin{array}{ccc}
                    (t\otimes u)\otimes v&=&t\otimes(u\otimes v)\\[.4em]
                    \tinytf{parassociativity_left}&=&\tinytf{parassociativity_right}
                \end{array}
            \end{equation}
        \end{subfigure}

        \vspace{-0.5em}
        \begin{subfigure}{0.30\textwidth}
            \begin{equation}\label{eq:idinduct}
                \begin{array}{ccc}
                    \textup{id}_{n+m}&=&\textup{id}_{n}\otimes\textup{id}_{m}\\[.4em]
                    \tinytf{idinduct_left}&=&\tinytf{idinduct_right}
                \end{array}
            \end{equation}
        \end{subfigure}\hspace{2em}
        \begin{subfigure}{0.40\textwidth}
            \begin{equation}\label{eq:paridentity}
                \begin{array}{ccccc}
                    t\otimes \textup{id}_0&=&t&=&\textup{id}_0\otimes t\\[.4em]
                    \tinytf{paridentity_left}&=&\tinytf{paridentity_mid}&=&\tinytf{paridentity_right}
                \end{array}
            \end{equation}
        \end{subfigure}

        \vspace{-0.5em}
        \begin{subfigure}{0.57\textwidth}
            \begin{equation}\label{eq:interchange}
                \begin{array}{ccc}
                    (u_1\semicolon u_2)\otimes(v_1\semicolon v_2)&=&(u_1\otimes v_1)\semicolon(u_2\otimes v_2)\\[.4em]
                    \tinytf{interchange_left}&=&\tinytf{interchange_right}
                \end{array}
            \end{equation}
        \end{subfigure}\hspace{1em}
        \begin{subfigure}{0.36\textwidth}
            \begin{equation}\label{eq:emptysymmetry}
                \begin{array}{ccccc}
                    \sigma_{n,0}&=&\textup{id}_n&=&\sigma_{0,n}\\[.4em]
                    \tinytf{emptysymmetry_left}&=&\tinytf{emptysymmetry_mid}&=&\tinytf{emptysymmetry_right}
                \end{array}
            \end{equation}
        \end{subfigure}

        \vspace{-0.5em}
        \begin{subfigure}{0.34\textwidth}
            \begin{equation}\label{eq:involution}
                \begin{array}{ccc}
                    \sigma_{n,m}\semicolon\sigma_{m,n}&=&\textup{id}_{n+m}\\[.4em]
                    \tinytf{involution_left}&=&\tinytf{involution_right}
                \end{array}
            \end{equation}
        \end{subfigure}\hspace{2em}
        \begin{subfigure}{0.48\textwidth}
            \begin{equation}\label{eq:naturality}
                \begin{array}{ccc}
                    (t\otimes\textup{id}_k)\semicolon\sigma_{m,k}&=&\sigma_{n,k}\semicolon(\textup{id}_k\otimes t)\\[.4em]
                    \tinytf{naturality_left}&=&\tinytf{naturality_right}
                \end{array}
            \end{equation}
        \end{subfigure}

        \vspace{-0.5em}
        \begin{subfigure}{0.86\textwidth}
            \begin{equation}\label{eq:symmetryinduct}
                \begin{array}{ccc}
                    \sigma_{n+k,m+\ell}&=&(\textup{id}_n\otimes\sigma_{k,m}\otimes\textup{id}_\ell)\semicolon(\sigma_{n,m}\otimes\sigma_{k,\ell})\semicolon(\textup{id}_m\otimes\sigma_{n,\ell}\otimes\textup{id}_k)\\[.4em]
                    \tf{symmetryinduct_left}&=&\tf{symmetryinduct_right}
                \end{array}
            \end{equation}
        \end{subfigure}
        \vspace{0.2em}
    \end{center}\end{minipage}}
    \caption{\label{fig:coherencelaws} \cref{eq:seqidentity,eq:seqassociativity,eq:parassociativity,eq:idinduct,eq:paridentity,eq:interchange} are the coherence equations of PROs. \cref{eq:seqidentity,eq:seqassociativity,eq:parassociativity,eq:idinduct,eq:paridentity,eq:interchange,eq:emptysymmetry,eq:involution,eq:naturality,eq:symmetryinduct} are the coherence equations of PROPs. Dotted boxes and wires highlight the correspondence with the syntax. For convenience, we can draw a single wire to represent multiple wires when it is clear from the context.}
\end{figure}

The PRO freely generated by the set of generators $\Sigma$, denoted $\cat{PRO}_\Sigma$, is defined as the collection of equivalence classes of terms $\terms{\Sigma}$ under \cref{eq:seqidentity,eq:seqassociativity,eq:parassociativity,eq:idinduct,eq:paridentity,eq:interchange}. Furthermore, given a set of equations $\mathcal{E}$ between terms, let $\cat{PRO}_{\Sigma,\mathcal{E}}$ be the PRO whose equivalence classes are the ones of $\cat{PRO}_\Sigma$ additionally quotiented by the equations of $\mathcal{E}$. The equivalence relation is denoted $\cat{PRO}_{\Sigma,\mathcal{E}}\vdash t_1=t_2$ when $t_1$ can be rewritten into $t_2$ using \cref{eq:seqidentity,eq:seqassociativity,eq:parassociativity,eq:idinduct,eq:paridentity,eq:interchange} together with the equations in $\mathcal{E}$. The equivalence relation is closed under context, meaning that if $\cat{PRO}_{\Sigma,\mathcal{E}}\vdash t_1=t_2$ where $t_i:r\to s$, then $\cat{PRO}_{\Sigma,\mathcal{E}}\vdash K[t_1]=K[t_2]$ for any contextualized term $K[\cdot]\in\contexts{\Sigma}{r}{s}$.

The PROP freely generated by the set of generators $\Sigma$, denoted $\cat{PROP}_\Sigma$, is defined as $\cat{PROP}_\Sigma\defeq\cat{PRO}_{\Sigma_0\cup\Sigma,\mathcal{E}_0}$ where $\mathcal{E}_0$ is the set containing \cref{eq:emptysymmetry,eq:involution,eq:naturality,eq:symmetryinduct}. That is, the terms are induced by $\Sigma$ as well as the set of swaps $\Sigma_0$, and are quotiented by \cref{eq:seqidentity,eq:seqassociativity,eq:parassociativity,eq:idinduct,eq:paridentity,eq:interchange,eq:emptysymmetry,eq:involution,eq:naturality,eq:symmetryinduct} to form the equivalence class of $\cat{PROP}_\Sigma$. Intuitively, compared to PROs, PROPs can bend the wires and gates can slide along those bends.

In this paper, we focus on two diagrammatic equivalence problems. For the first one (\cref{sec:pro}), we consider an arbitrary \emph{state-and-effect-free} PRO, meaning that the set of generators is arbitrary but contains no states and no effects. In the second one (\cref{sec:perm}), we consider the PRO of permutations $\cat{PRO}_{\Sigma_0,\mathcal{E}_0}$, which is precisely the simplest PROP that one can think about, namely $\cat{PROP}_\varnothing$.

\subsection{Notations}


A list of $p$ elements of a set $A$ is denoted $xs=[xs_1,\dots,xs_p]\in A^p$. The list with $0$ elements is denoted $\emptylist\in A^0$. One can plug $x\in A$ to the left of $xs\in A^p$ with the operator $\first$ defined as $x\first xs\defeq [x,xs_1, \ldots ,xs_p]\in A^{p+1}$. More generally, one can define the concatenation of two lists $xs\in A^p$ and $ys\in A^q$ with the operator $\conc$ defined as $xs\conc ys\defeq [xs_1,\ldots, xs_p, ys_1, \ldots ,ys_q]\in A^{p+q}$. If $(A,+)$ is a group, given $xs\in A^p$ and $x\in A$, we define $xs\listplus x\defeq[xs_1+x,\dots,xs_p+x]$ and $xs\listminus x\defeq xs\listplus (-x)$.

Term rewriting systems (TRS) are defined as usual \cite{trs}. In particular, given two terms $t_1$ and $t_2$, we write $t_1\rewrite{\mathcal{R}} t_2$ when $t_1$ can be rewritten into $t_2$ using exactly one step of the rewriting system $\mathcal{R}$, and $t_1\mrewrite{\mathcal{R}} t_2$ when $t_1$ can be rewritten into $t_2$ using any finite number (including zero) of steps of $\mathcal{R}$. We write $t\notrewrite{\mathcal{R}}$ when no rules of $\mathcal{R}$ can be applied from the term $t$.


\section{Term Rewriting in a State-and-effect-free PRO}\label{sec:pro}

In this section we fix $\Sigma$ a state-and-effect-free set of generators, meaning that each generator $g\in\Sigma$ satisfies $\textup{dom}(g)>0$ and $\textup{cod}(g)>0$. This assumption enables termination and confluence of our TRS while remaining sufficiently general for a wide range of practical applications. We consider $\cat{PRO}_\Sigma$ and solve its underlying diagrammatic equivalence problem: given two terms $t_1,t_2\in\terms{\Sigma}$ such that $\cat{PRO}_\Sigma\vdash t_1=t_2$, how to effectively rewrite $t_1$ into $t_2$ using the coherence equations? To do so, we introduce a normal form and an associated rewriting system that we prove to be terminating and confluent.

\subsection{Normal Form and Preprocessing}

We define a \emph{pre-normal form} as a sequence of layers, a layer being a single generator $g$ sandwiched between identities: $\layer{k}{\ell}{g} \defeq \textup{id}_{k}\otimes (g\otimes\textup{id}_{\ell})$ where $k,\ell\in \mathbb N$ and $g\in \Sigma$. We say that $g$ has  \emph{upper width}  $k$, \emph{lower width} $\ell$, and \emph{height} $k+\text{dom}(g)$. 

A term $\layers{ks}{\ell s}{gs}$, defined as follows, is a non-empty sequence of $p>0$ layers, characterized by the list $gs\in\Sigma^p$  of the associated generators and the lists $ks,\ell s\in\N^p$ of their corresponding upper  and lower widths. 
\begin{gather*}
    \layers{k\first ks}{\ell\first \ell s}{g\first gs} \defeq \begin{cases} 
        \layer{k}{\ell}{g} &\text{if $ks=\ell s=gs = \emptylist$}\\
        \layer{k}{\ell}{g}\semicolon \layers{ks}{\ell s}{gs} &\text{otherwise}
    \end{cases}
\end{gather*}

Notice that to guarantee that $\layers{ks}{\ell s}{gs}$ is well-defined, we assume throughout this paper that the three lists $ks$, $\ell s$ and $gs$ have the same size $p$ and that $\forall i\in \natset{p-1}$, $ks_i+\textup{cod}(gs_i) + \ell s_i = ks_{i+1}+\textup{dom}(gs_{i+1}) + \ell s_{i+1}$, to ensure that the layers have appropriate domains and codomains. 

Notice that, according to the coherence equations, two layers in sequence may commute when the outputs of the first generator are not connected to any inputs of the second generator. To obtain a \emph{normal form}, we must fix an order for the commuting layers. To do so, we require the upper width of a generator to be smaller than the height of the next one, since the two layers can be commuted otherwise. This leads to the following definition of normal forms.

\begin{definition}[normal form]\label{def:nf}
    The set $\nf{\Sigma}$ of terms in \emph{normal form} is inductively defined as follows: $\textup{id}_n$ and $\layer{k}{\ell}{g}$ are in normal form ; and $\layers{k\first ks}{\ell\first \ell s}{g\first gs}$ is in normal form if $\layers{ks}{\ell s}{gs}$ is in normal form while satisfying $k<ks_1+\textup{dom}(gs_1)$.
\end{definition}

\begin{example} 
    For instance, if $\Sigma=\{A:1\to1,B:2\to2\}$, the term $t=((A\semicolon\textup{id}_1)\otimes ((A\otimes A)\semicolon B))\semicolon(B\otimes A)\in\terms{\Sigma}$ is equivalent to $\layers{[0,1,2,1,0,2]}{[2,1,0,0,1,0]}{[A,A,A,B,B,A]}$ which is in normal form. These terms can be depicted as follows.
    \begin{gather*}
        \tf{nf_example_right} = \tf{nf_example_left}
    \end{gather*}
\end{example}

We introduce a procedure to turn any term into its normal form, relying on two new operators $\bullet$ and $\star$ allowing to compute the resulting normal form of a sequential and a parallel composition of two normal forms respectively. While the definition of $\star$ is straightforward, the definition of $\bullet$ requires mixing the layers of the two normal forms in order to satisfy the conditions explained in \cref{def:nf}.
\begin{definition}[normalizing map]\label{def:nfmap}
    Let ${\textup{NF}:\terms{\Sigma}\to \nf{\Sigma}}$ be the \emph{normalizing map} that assigns to each term $t:n\to m$ its normal form $\textup{NF}(t) : n\to m$, inductively defined as
    \vspace{-0.5em}
    \begin{gather*}
        \textup{NF}(\textup{id}_n) \defeq \textup{id}_n \hspace{3em}
        \textup{NF}(g) \defeq \layer{0}{0}{g} \\
        \textup{NF}(u\semicolon v) \defeq \textup{NF}(u)\bullet\textup{NF}(v) \hspace{3em}
        \textup{NF}(u\otimes v) \defeq \textup{NF}(u)\star\textup{NF}(v)
    \end{gather*}
    where $\bullet$ and $\star$ are inductively defined as follows, where $c=\textup{cod}(\layers{ks}{\ell s}{gs})$ and $d = \textup{dom}(\layers{rs}{ss}{hs})$.
    \begin{align*}
        \textup{id}_n \bullet \textup{id}_n&\defeq  \textup{id}_n \\[0.2em]
        \textup{id}_n \bullet \layers{ks}{\ell s}{gs} &\defeq \layers{ks}{\ell s}{gs} \\[0.2em]
        \layers{ks}{\ell s}{gs}  \bullet  \textup{id}_n&\defeq \layers{ks}{\ell s}{gs} \\[0.2em]
        \layer{k'}{\ell'}{g'}\bullet \layers{k\first ks}{\ell\first \ell s}{g\first gs}&\defeq \begin{cases}
            \layers{k'\first k\first ks}{\ell'\first \ell\first \ell s}{g'\first g\first gs} \;\text{if $k'<k+\textup{dom}(g)$}\\
            \layer{k}{\ell+\textup{def}(g')}{g}\semicolon \layer{k'{-}\textup{def}(g)}{\ell'}{g'} \; \text{else if}\; ks=\ell s=gs=\emptylist\\
            \layer{k}{\ell+\textup{def}(g')}{g}\semicolon (\layer{k'{-}\textup{def}(g)}{\ell'}{g'}\bullet \layers{ks}{\ell s}{gs}) \; \text{otherwise}
        \end{cases}\\[0.2em]
        \layers{k\first ks}{\ell \first \ell s}{g \first gs}\bullet \layers{rs}{ss}{hs}&\defeq \layer{k}{\ell}{g}\bullet   (\layers{ks}{\ell s}{gs}\bullet \layers{rs}{ss}{hs})\\[0.2em]
        \textup{id}_n \star \textup{id}_m &\defeq \textup{id}_{n+m}\\[0.2em]
        \textup{id}_n \star \layers{ks}{\ell s}{gs} &\defeq \layers{ks\listplus n}{\ell s}{gs} \\[0.2em]
        \layers{ks}{\ell s}{gs}  \star  \textup{id}_n&\defeq \layers{ks}{\ell s\listplus n}{gs} \\[0.2em]
        \layers{ks}{\ell s}{gs}\star \layers{rs}{ss}{hs}&\defeq \layers{ks\smallconc (rs\listplus c)}{(\ell s\listplus d )\smallconc  ss}{gs \conc hs}
    \end{align*}
\end{definition}

Graphically, the notion of diagrammatic equivalence in PROs is known as \emph{planar isotopy} \cite[Theorem~3.1]{selingerbible}, and, roughly speaking, means that boxes can be moved horizontally without allowing wires to cross each other. To avoid graph-theoretic technicalities and keep the focus on term rewriting, we work under the assumption that the normalizing map $\textup{NF}$ yields exactly the notion of diagrammatic equivalence in PROs without showing that it relates to planar isotopy.

As stated in the following lemma, the normalizing map's definition ensures that all normal forms are fixed points, meaning that normalizing a normal form does not change the normal form.
\begin{lemma}\label{lem:nfnfterm}
    $\textup{NF}(t)=t$ for any $t\in \nf{\Sigma}$.
\end{lemma}

\begin{proof} The proof is by induction on the number of layers of $t$. The statement is true when $t=\textup{id}_n$ and $t=\layer{k}{\ell}{g}$. Assume $\layers{k\first ks}{\ell \first \ell s}{g\first gs}$ is in normal form, so in particular $k<ks_1+\textup{dom}(gs_1)$, which implies the following.
\begin{multline*}
    \textup{NF}(\layers{k\first ks}{\ell \first \ell s}{g\first gs}) = \textup{NF}(\layer{k}{\ell}{g})\bullet \textup{NF}(\layers{ks}{\ell s}{gs}) \\
    \overset{\textup{IH}}{=} \layer{k}{\ell}{g}\bullet \layers{ks}{\ell s}{gs} = \layers{k\first ks}{\ell \first \ell s} {g\first gs}
\end{multline*}
This proof has been verified with Isabelle/HOL. \qed
\end{proof}

In order to simplify the proofs on the TRS and make the correspondence between generators and layers clearer, we apply a preprocessing step to terms that replaces each generator $g$ by its corresponding layer $\layer{0}{0}{g}$. The set of preprocessed terms is denoted $\pp{\Sigma}$ and satisfies $\nf{\Sigma}\subset\pp{\Sigma}\subset\terms{\Sigma}$.

\subsection{Terminating and Confluent Rewriting System}

Let $\mathcal{R}$ be the TRS containing the rewrite rules \labelcref{ru:seqasso,ru:parasso,ru:seqidleft,ru:seqidright,ru:idid,ru:ididctx,ru:par,ru:parseqidleft,ru:parseqidright,ru:layercom,ru:layercomctx} defined in \cref{fig:prors}. Rewriting can be applied to sub-terms, meaning that $K[L]\rewrite{\mathcal{R}}K[R]$ is a valid rewrite for any rewrite rule $(L\rewrite{}R)\in\mathcal{R}$ and any $K[\cdot]\in\ppcontexts{\Sigma}{\textup{dom}(L)}{\textup{cod}(L)}$. Moreover, some rewrite rules (\labelcref{ru:par,,ru:layercom,,ru:layercomctx}) are subject to additional applicability conditions. All such conditions are decidable and stable under rewriting, ensuring that the normalization procedure is effective and suitable for automation. As stated in the following proposition, the rewriting system $\mathcal{R}$ is sound with respect to the coherence equations.
\begin{proposition}\label{prop:prosoundness}
    $t_1\mrewrite{\mathcal{R}}t_2\implies\cat{PRO}_\Sigma\vdash t_1=t_2$ for any $t_1,t_2\in\pp{\Sigma}$.
\end{proposition}
\begin{proof}
    For any rewrite rule $(L\rewrite{}R)\in\mathcal{R}$ and any $K[\cdot]\in\ppcontexts{\Sigma}{\textup{dom}(L)}{\textup{cod}(L)}$, we prove $\cat{PRO}_\Sigma\vdash K[L]=K[R]$ by induction on $K[\cdot]$ where the only non-trivial case is $K[\cdot]=[\cdot]$. For most rewrite rules this is trivial; the only interesting cases \labelcref{ru:layercom,ru:layercomctx} are proved using, in particular, \cref{eq:interchange}. \qed
\end{proof}

\begin{figure}[t]
    \fbox{\begin{minipage}{0.982\linewidth}\begin{center}
        \vspace{-1em}
        \hspace{-0.8em}
        \begin{subfigure}{0.37\textwidth}
            \begin{equation}\label{ru:seqasso}\tag{R1}
                (t\semicolon u)\semicolon v \longrightarrow t\semicolon (u\semicolon v)
            \end{equation}
        \end{subfigure}\hspace{2em}
        \begin{subfigure}{0.42\textwidth}
            \begin{equation}\label{ru:parasso}\tag{R2}
                (t\otimes u)\otimes v \longrightarrow t\otimes (u\otimes v)
            \end{equation}
        \end{subfigure}

        \vspace{-0.3em}
        \hspace{-0.8em}
        \begin{subfigure}{0.25\textwidth}
            \begin{equation}\label{ru:seqidleft}\tag{R3}
                \textup{id}_n\semicolon t \longrightarrow t
            \end{equation}
        \end{subfigure}\hspace{2em}
        \begin{subfigure}{0.25\textwidth}
            \begin{equation}\label{ru:seqidright}\tag{R4}
                t\semicolon\textup{id}_n \longrightarrow t
            \end{equation}
        \end{subfigure}

        \vspace{-0.3em}
        \hspace{-0.8em}
        \begin{subfigure}{0.36\textwidth}
            \begin{equation}\label{ru:idid}\tag{R5}
                \textup{id}_n\otimes\textup{id}_m \longrightarrow \textup{id}_{n+m}
            \end{equation}
        \end{subfigure}\hspace{2em}
        \begin{subfigure}{0.46\textwidth}
            \begin{equation}\label{ru:ididctx}\tag{R6}
                \textup{id}_n\otimes(\textup{id}_m\otimes t) \longrightarrow \textup{id}_{n+m}\otimes t
            \end{equation}
        \end{subfigure}

        \vspace{-0.3em}
        \hspace{-0.8em}
        \begin{subfigure}{0.82\textwidth}
            \begin{equation}\label{ru:par}\tag{R7}
                u\otimes v \longrightarrow (u\otimes\textup{id}_{\textup{dom}(v)})\semicolon(\textup{id}_{\textup{cod}(u)}\otimes v)
                \;\;\text{if}\;\; u\ne\textup{id}_n\wedge v\ne\textup{id}_m
            \end{equation}
        \end{subfigure}

        \vspace{-0.3em}
        \hspace{-0.8em}
        \begin{subfigure}{0.54\textwidth}
            \begin{equation}\label{ru:parseqidleft}\tag{R8}
                \textup{id}_n\otimes(u\semicolon v) \longrightarrow (\textup{id}_n\otimes u)\semicolon(\textup{id}_n\otimes v)
            \end{equation}
        \end{subfigure}

        \vspace{-0.3em}
        \hspace{-0.8em}
        \begin{subfigure}{0.54\textwidth}
            \begin{equation}\label{ru:parseqidright}\tag{R9}
                (u\semicolon v)\otimes\textup{id}_n \longrightarrow (u\otimes\textup{id}_n)\semicolon(v\otimes\textup{id}_n)
            \end{equation}
        \end{subfigure}

        \vspace{-0.3em}
        \hspace{-0.8em}
        \begin{subfigure}{0.73\textwidth}
            \begin{equation}\label{ru:layercom}\tag{R10}
                \begin{array}{c}
                    \layer{k_1}{\ell_1}{g_1}\semicolon\layer{k_2}{\ell_2}{g_2} \longrightarrow \layer{k_2}{\ell_2+\textup{def}(g_1)}{g_2}\semicolon \layer{k_1-\textup{def}(g_2)}{\ell_1}{g_1} \\
                    \text{if}\;\; k_1\ge k_2+\textup{dom}(g_2)
                \end{array}
            \end{equation}
        \end{subfigure}

        \vspace{-0.3em}
        \hspace{-0.8em}
        \begin{subfigure}{0.85\textwidth}
            \begin{equation}\label{ru:layercomctx}\tag{R11}
                \begin{array}{c}
                    \layer{k_1}{\ell_1}{g_1}\semicolon (\layer{k_2}{\ell_2}{g_2}\semicolon t) \longrightarrow \layer{k_2}{\ell_2+\textup{def}(g_1)}{g_2}\semicolon (\layer{k_1-\textup{def}(g_2)}{\ell_1}{g_1}\semicolon t) \\
                    \text{if}\;\; k_1\ge k_2+\textup{dom}(g_2)
                \end{array}
            \end{equation}
        \end{subfigure}
        \vspace{0.2em}
    \end{center}\end{minipage}}
    \caption{The rewrite rules are defined for any $n,m,k_1,k_2,\ell_1,\ell_2\in\N$, $g_1,g_2\in\Sigma$, $t,u,v\in\terms{\Sigma}$ and \labelcref{ru:par,,ru:layercom,,ru:layercomctx} are conditioned. \label{fig:prors}}
\end{figure}

In the following, we prove that preprocessed terms are reduced into their normal form by applying $\mathcal{R}$. To do so, we show that $\mathcal{R}$ terminates and is confluent. We initially attempted to rely on automated termination and confluence provers such as \textsf{AProVE}~\cite{aprove} and the provers of the \textsf{CoCo}{} competition~\cite{coco}. While these tools were useful for guiding early refinements of the rules, they ultimately did not succeed in producing termination or confluence certificates for our system (at least when used via their web interface, given our lack of expertise). We therefore carried out the proofs manually and formalized the most important steps with the proof assistant Isabelle/HOL. To keep the formalization lightweight, we chose not to use the term rewriting infrastructure present in Isabelle/HOL, concentrating instead our formalizing efforts on the equalities and inequalities at the heart of the proofs.

\begin{proposition}\label{prop:proterminaison}
    $\mathcal{R}$ terminates.
\end{proposition}
\begin{proof}
    We show that any rewrite sequence $t_1\rewrite{\mathcal{R}}t_2\rewrite{\mathcal{R}}t_3\rewrite{\mathcal{R}}\dots$ is finite. To do so, we define two positive quantities $\beta:\terms{\Sigma}\to\N$ and $\delta_D:\terms{\Sigma}\to\R^+$ (see \cref{fig:quantities}), where the latter is parameterized by $D\in\N$, and prove that $(\beta,\delta_D)$ strictly decreases by lexicographical order, for a carefully chosen $D\in\N$. 
    Notice that as $\delta_D$ takes real positive values, we must also show that it decreases by at least a constant at each step, otherwise the value could decrease infinitely.   
    Given a term $t\in\terms{\Sigma}$, let $D(t)\defeq 2+\max_{g\in\textup{gen}(t)}|\textup{def}(g)|$. Notice that $D(\cdot)$ is preserved under rewriting, meaning that $D(t_i)=D(t_j)$ for any $i,j$. We prove that
    \begin{gather*}
        \beta(K[L])>\beta(K[R])
        \quad\text{or}\quad
        \beta(K[L])=\beta(K[R])\wedge\delta_D(K[L])\ge\delta_D(K[R])+\nicefrac{1}{D}
    \end{gather*}
    for any rewrite rule $(L\rewrite{}R)\in\mathcal{R}$ and any context $K[\cdot]\in\ppcontexts{\Sigma}{\textup{dom}(L)}{\textup{cod}(L)}$, where we chose $D\defeq D(K[L])=D(K[R])$. This is straightforwardly proved by induction on $K[\cdot]$ where the base cases $K[\cdot]=[\cdot]$ have been verified with Isabelle/HOL. More precisely, we prove  that \labelcref{ru:seqasso,ru:parasso,ru:seqidleft,ru:seqidright,ru:idid,ru:ididctx,ru:par,ru:parseqidleft,ru:parseqidright} strictly decrease $\beta$ while \labelcref{ru:layercom,ru:layercomctx} preserve $\beta$ and strictly decrease $\delta_D$. \qed
\end{proof}

\begin{remark}
    The fact that \labelcref{ru:layercom,ru:layercomctx} strictly decrease $\delta_D$ is not trivial, and is only true with a careful choice of $D\in\N$. More precisely, $D\defeq D(K[L])=D(K[R])$ ensures that $D\ge 2-\textup{def}(g_2)$, which implies $1-\nicefrac{1}{D}\ge\nicefrac{1}{D}(1-\textup{def}(g_2))$. Moreover, the condition $k_1\ge k_2+\textup{dom}(g_2)$ with $\textup{dom}(g_2)\ge1$ yields the following.
    \begin{gather*}
        k_1\ge k_2+1
        \implies (1-\nicefrac{1}{D})k_1\ge (1-\nicefrac{1}{D})(k_2+1) \\
        \implies k_1 + \nicefrac{1}{D}k_2 \ge k_2 + \nicefrac{1}{D}k_1 + (1-\nicefrac{1}{D}) \ge k_2 + \nicefrac{1}{D}k_1 + \nicefrac{1}{D}(1-\textup{def}(g_2)) \\
        \implies k_1 + \nicefrac{1}{D}k_2 \ge k_2 + \nicefrac{1}{D}(k_1-\textup{def}(g_2)) + \nicefrac{1}{D}
        \implies \delta_D(L)\ge\delta_D(R)+\nicefrac{1}{D}
    \end{gather*}
\end{remark}

\begin{lemma}\label{lem:nfispreserved}
    $t_1\rewrite{\mathcal{R}}^* t_2{\implies}\textup{NF}(t_1)=\textup{NF}(t_2)$ for any $t_1,t_2\in\pp{\Sigma}$.
\end{lemma}
\begin{proof}
    For any rewrite rule $(L\rewrite{}R)\in\mathcal{R}$ and any $K[\cdot]\in\ppcontexts{\Sigma}{\textup{dom}(L)}{\textup{cod}(L)}$, the equality $\textup{NF}(K[L])=\textup{NF}(K[R])$ is straightforwardly proved by induction on $K[\cdot]$ where the base cases $K[\cdot]=[\cdot]$ have been verified with Isabelle/HOL. \qed
\end{proof}

\begin{figure}[t]
    \begin{center}
        \setlength{\tabcolsep}{2.4pt}
        \begin{tabular}{|c|c|c|c|c|}
            \hline
            & $\alpha:\terms{\Sigma}\to\N$ & $\beta:\terms{\Sigma}\to\N$ & $\gamma:\terms{\Sigma}\to\N$ & $\delta_D:\terms{\Sigma}\to\R^+$ \\
            \hline
            $\textup{id}_n$ & $0$ & $2$ & $n$ & $0$ \\
            \hline
            $g$ & $\begin{array}{c}
                nm^2{+}1 \;\text{if}\; g=\sigma_{n,m}\\
                0 \;\text{otherwise}
            \end{array}$ & $5$ & $0$ & $0$ \\
            \hline
            $u\semicolon v$ & $\alpha(u)+\alpha(v)$ & $2\beta(u)+\beta(v)+1$ & $\gamma(u)+\gamma(v)$ & $\delta_D(u)+\frac{1}{D}\delta_D(v)$ \\ 
            \hline
            $u\otimes v$ & $\alpha(u)+\alpha(v)$ & $\beta(u)^2\beta(v)$ & $\gamma(u)+2\gamma(v)$ & $\begin{array}{c}
                n+\delta_D(v) \;\text{if}\; u=\textup{id}_n\\
                \delta_D(u)+\delta_D(v) \;\text{otherwise}
            \end{array}$\\
            \hline
            \hline
            $\layer{k}{\ell}{g}$ & $\begin{array}{c}
                d+1 \;\text{if}\; g=\tau_d\\
                \alpha(g) \;\text{otherwise}
            \end{array}$ & $200$ & $k+4\ell$ & $k$ \\
            \hline
        \end{tabular}

        \vspace{1em}

        \setlength{\tabcolsep}{4.95pt}
        \begin{tabular}{|c|c|c|c|c|}
            \hline
            Rewrite rules & $\;\;\alpha\;\;$ & $\;\;\beta\;\;$ & $\;\;\gamma\;\;$ & $\;\;\delta_D\;\;$ \\
            \hline
            \labelcref{ru:seqasso,,ru:parasso,,ru:seqidleft,,ru:seqidright,,ru:idid,,ru:ididctx,,ru:par,,ru:parseqidleft,,ru:parseqidright} & $=$ & $>$ & & \\
            \hline
            \labelcref{ru:layercom,ru:layercomctx} & & $=$ & & $>$ \\
            \hline
            \labelcref{ru:swap:swap0left,,ru:swap:swap0right,,ru:swap:swapreduce} & $>$ & & & \\
            \hline
            \labelcref{ru:swap:layernat,ru:swap:layernatctx} & $=$ & $=$ & $>$ & \\
            \hline
            \labelcref{ru:swap:layercon,ru:swap:layerconctx} & $>$ & & & \\
            \hline
            \labelcref{ru:swap:layerfus,ru:swap:layerfusctx} & $>$ & & & \\
            \hline
            \labelcref{ru:swap:layercom,ru:swap:layercomctx} & $=$ & $=$ & $=$ & $>$ \\
            \hline
        \end{tabular}
    \end{center}
    \caption{The top table defines the positive quantities $\alpha,\beta,\gamma,\delta_D$, where the former is parameterized by $D\in\N$. These quantities are used to prove the termination of $\mathcal{R}$ and $\mathcal{R}_0$ in \cref{prop:proterminaison,prop:permterminaison} respectively. The last line of the table gives the value of each quantity applied to layers. The bottom table describes which rewrite rules strictly decrease ($>$) or preserve ($=$) each quantity. All cases so indicated have been verified with Isabelle/HOL. \label{fig:quantities}}
\end{figure}

\begin{lemma}\label{lem:nfterminaison}
    $t\notrewrite{\mathcal{R}}\implies t\in \nf{\Sigma}$ for any $t\in \pp{\Sigma}$.
\end{lemma}
\begin{proof}
    By induction on $t$. The base cases $t=\textup{id}_n$ for $n\in\N$ are direct because $\textup{id}_n\notrewrite{\mathcal{R}}$ and $\textup{id}_n\in \nf{\Sigma}$. The base cases $t\in\Sigma$ are trivial because $t\notin \pp{\Sigma}$.
    \begin{itemize}
        \item Suppose $t=u\semicolon v$ and $t\notrewrite{\mathcal{R}}$. Then $u,v\in \pp{\Sigma}$ and $u\notrewrite{\mathcal{R}},v\notrewrite{\mathcal{R}}$ which, by IH, implies $u,v\in \nf{\Sigma}$. Notice that $u\ne\textup{id}_n\ne v$ otherwise \labelcrefor{ru:seqidleft,ru:seqidright} applies. So, it must be the case that $u=\layer{k}{\ell}{g}$ and $v=\layers{ks}{\ell s}{gs}$ for some $k,\ell,p\in\N$, $g\in\Sigma$, $ks,\ell s\in\N^p$, $gs\in\Sigma^p$ such that $p>0$, otherwise \eqref{ru:seqasso} applies. Moreover, $k<ks_1+\textup{dom}(gs_1)$, otherwise \labelcrefor{ru:layercom,ru:layercomctx} applies. Hence, $t=\layers{k\first ks}{\ell\first \ell s}{g\first gs}\in \nf{\Sigma}$.
        \item Suppose $t=u\otimes v$ and $t\notrewrite{\mathcal{R}}$. Then $u\notrewrite{\mathcal{R}},v\notrewrite{\mathcal{R}}$. Notice that $u$ and $v$ are not necessarily in $\pp{\Sigma}$, so we cannot use the IH. Nevertheless, we know that $u,v\notin\Sigma$, otherwise $t\notin \pp{\Sigma}$. In fact, we know for sure that $u=\textup{id}_k$ and $v=v_1\otimes v_2$ for some $k\in\N$ and $v_1,v_2\in\terms{\Sigma}$, otherwise a rule among \labelcref{ru:parasso,,ru:idid,,ru:par,,ru:parseqidleft,,ru:parseqidright} applies. For the same reason and the fact that \eqref{ru:ididctx} is not applicable, it must be the case that $v_1=g$ and $v_2=\textup{id}_\ell$ for some $g\in\Sigma$ and $\ell\in\N$. Hence, $t=\textup{id}_k\otimes(g\otimes\textup{id}_\ell)=\layer{k}{\ell}{g}\in \nf{\Sigma}$.
    \end{itemize} \qed
\end{proof}

\begin{proposition}\label{prop:proconfluence}
    $\mathcal{R}$ is confluent in $\pp{\Sigma}$.
\end{proposition}
\begin{proof}
    Let $t\in \pp{\Sigma}$. If $t\rewrite{\mathcal{R}}^*u$ and $t\rewrite{\mathcal{R}}^*v$, then by \cref{prop:proterminaison} there exist $u',v'\in\pp{\Sigma}$ such that $u\rewrite{\mathcal{R}}^*u'\notrewrite{\mathcal{R}}$ and $v\rewrite{\mathcal{R}}^*v'\notrewrite{\mathcal{R}}$. \cref{lem:nfnfterm,lem:nfispreserved,lem:nfterminaison} imply $\textup{NF}(t)=u'=v'$, which concludes the proof. \qed
\end{proof}

As stated in the following theorem, the termination and confluence of $\mathcal{R}$ solve our problem: given two diagrammatically equivalent terms $t_1,t_2\in\pp{\Sigma}$, we can effectively rewrite $t_1$ into $t_2$ using the coherence equations.

\begin{theorem}
    $\textup{NF}(t_1)=\textup{NF}(t_2)\implies t_1\bothmrewrite{\mathcal{R}} t_2$ for any $t_1,t_2\in\pp{\Sigma}$.
\end{theorem}
\begin{proof}
    \cref{prop:proterminaison,lem:nfnfterm,lem:nfispreserved,lem:nfterminaison} imply that $t_i\rewrite{\mathcal{R}}^*\textup{NF}(t_i)$, and the assumption $\textup{NF}(t_1)=\textup{NF}(t_2)$ yields $t_1\bothmrewrite{\mathcal{R}} t_2$.
    \qed
\end{proof}

While the rewrite rules can be applied in any order, it may be more efficient to follow a predefined rewriting strategy. In particular, \labelcref{ru:seqasso,ru:parasso,ru:seqidleft,ru:seqidright,ru:idid,ru:ididctx,ru:par,ru:parseqidleft,ru:parseqidright} allow transforming any preprocessed term into a pre-normal form, and \labelcref{ru:layercom,ru:layercomctx} allow transforming the layers to fulfill the normal form conditions. Hence, the two steps can be applied one after the other.

Notice that the normal form map $\textup{NF}$ is useful on its own.
Indeed, checking diagrammatic equivalence between two terms $t_1,t_2\in\pp{\Sigma}$ can be directly achieved by checking $\textup{NF}(t_1)=\textup{NF}(t_2)$.
Nevertheless, as we plan to integrate our TRSs in an overall rewriting-based certification pipeline, the TRS is more interesting to us than the normal form itself.


\section{Term Rewriting in the PRO of Permutations}\label{sec:perm}

In this section, we consider the PRO of permutations $\cat{PRO}_{\Sigma_0,\mathcal{E}_0}$, whose generators are the swaps, i.e. $\Sigma_0\defeq\{\sigma_{n,m}:n,m\in\N\}$ and where $\mathcal{E}_0$ is the set containing \cref{eq:emptysymmetry,eq:involution,eq:naturality,eq:symmetryinduct}. $\cat{PRO}_{\Sigma_0,\mathcal{E}_0}$ is the simplest PROP that one can think about, namely $\cat{PROP}_\varnothing$. Intuitively, in this structure, the diagrams have no boxes and are just made of wires, which allows expressing any permutation. Notice that all terms in $\terms{\Sigma_0}$ have an equal number of inputs and outputs. We solve the underlying diagrammatic equivalence problem: given two terms $t_1,t_2\in\terms{\Sigma_0}$ expressing the same permutation, how to effectively rewrite $t_1$ into $t_2$ using the coherence equations (\cref{eq:seqidentity,eq:seqassociativity,eq:parassociativity,eq:idinduct,eq:paridentity,eq:interchange,eq:emptysymmetry,eq:involution,eq:naturality,eq:symmetryinduct})? To do so, we refine the normal form of \cref{def:nf} to define a normal form for the PRO of permutations---that we call \emph{canonical form} to prevent confusion---and provide an associated rewriting system that we prove to be terminating and confluent.

\subsection{Canonical Form and Interpretation}

The permutations of size $n\in\N$, whose set is denoted $\cat{Perm}(n)$, are the bijective maps from $\natset{n}$ to itself. A permutation ${\pi\in\cat{Perm}(n)}$ is denoted $(\pi(1),\dots,\pi(n))$. For instance, ${(2,1,3)\in\cat{Perm}(3)}$ is the permutation mapping $1$ to $2$, $2$ to $1$ and $3$ to itself. There is only one permutation of size $0$, denoted $(\,)$, and only one permutation of size $1$, denoted $(1)$. In the following, $|\pi|$ denotes the size of the permutation $\pi$ and $\pi^{-1}$ its inverse. Two permutations $\pi_1,\pi_2\in\cat{Perm}(n)$ can be composed sequentially with $\circ$ as the usual composition, i.e.~$(\pi_2\circ\pi_1)(i)\defeq\pi_2(\pi_1(i))$ for any $i\in\natset{n}$. Two permutations $\pi_1\in\cat{Perm}(n),\pi_2\in\cat{Perm}(m)$ can be composed in parallel with the direct sum $\oplus$ defined as follows for any $i\in\natset{n+m}$.
\begin{gather*}
    (\pi_1\oplus\pi_2)(i)\defeq\begin{cases}
        \pi_1(i) &\text{if}\; 1\le i\le n\\
        n+\pi_2(i-n) &\text{if}\; n<i\le n+m
    \end{cases}
\end{gather*}

Let $\cat{Perm}\defeq\cup_{n\in\N}\cat{Perm}(n)$ be the collection of all permutations. Then, each term $t:n\to n\in\terms{\Sigma_0}$ can be interpreted as a permutation of size $n$.

\begin{definition}[interpretation]
    Let $\interp{\cdot}:\terms{\Sigma_0}\to\cat{Perm}$ be the \emph{interpretation}, inductively defined as follows.
    \begin{gather*}
        \interp{\textup{id}_n}\defeq (1,\dots,n)\hspace{3em}
        \interp{\sigma_{n,m}}\defeq (n+1,\dots,n+m,1,\dots,n)\\
        \interp{u\semicolon v}\defeq\interp{v}\circ\interp{u} \hspace{3em}
        \interp{u\otimes v}\defeq\interp{u}\oplus\interp{v}
    \end{gather*}
\end{definition}

Given $d\in\N$, let $\tau_d \defeq \sigma_{d,1}$ be the \emph{toboggan} of size $d$. Intuitively, the diagrammatic representation of $\tau_d$ lifts its last wire to the top, and its interpretation is $\interp{\tau_d}=(2,\dots,d+1,1)$, which implies the following equality.
\begin{gather*}
    \interp{\layer{k}{\ell}{\tau_d}}=(1,\dots,k,k+2,\dots,k+d+1,k+1,k+d+2,\dots,k+d+1+\ell)
\end{gather*}

We now define the \emph{canonical form} and its associated canonizing map. Given $ks,\ell s,ds\in\N^p$ with $p>0$, we introduce the following notation.
\begin{gather*}
    \clayers{ks}{\ell s}{ds}\defeq \layers{ks}{\ell s}{[\tau_{ds_i}]_{i\in\natset{p}}}
\end{gather*}

\begin{definition}[canonical form]\label{def:cf}
    The set $\cf{\Sigma_0}$ of terms in \emph{canonical form} is inductively defined as follows: $\textup{id}_n$ is in canonical form ; $\layer{k}{\ell}{\tau_d}$ is in canonical form if $d\ge1$ ; and $\clayers{k\first ks}{\ell\first \ell s}{d\first ds}$ is in canonical form if $d\ge1$ and $\clayers{ks}{\ell s}{ds}$ is in canonical form while satisfying $k<ks_1$.
\end{definition}

\begin{example}\label{ex:examplecf}
    The term $(\textup{id}_1\otimes\tau_1\otimes\tau_1)\semicolon(\sigma_{2,2}\otimes\textup{id}_1)\semicolon(\sigma_{1,2}\otimes\tau_1)\semicolon(\textup{id}_1\otimes\sigma_{1,2}\otimes\textup{id}_1)\semicolon(\textup{id}_3\otimes\tau_1)\semicolon(\tau_3\otimes\textup{id}_1)$ is equivalent to $\clayers{[0,1,2,3]}{[2,0,1,0]}{[2,3,1,1]}$, which is in canonical form. These terms can be graphically depicted as follows.
    \begin{gather*}
        \tf{cf_example_left} = \tf{cf_example_right}
    \end{gather*}
\end{example}

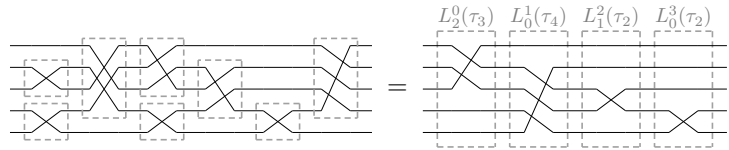
\begin{definition}[canonizing map]
    Let $\textup{CF}:\cat{Perm}\to \cf{\Sigma_0}$ be the \emph{canonizing map} that assigns to each permutation $\pi\in\cat{Perm}(n)$ its term in canonical form $\textup{CF}(\pi):n\to n$, inductively defined as
    \vspace{-0.5em}
    \begin{gather*}
        \textup{CF}((\;))\defeq \textup{id}_0 \hspace{3em}
        \textup{CF}((1))\defeq \textup{id}_1 \\
        \textup{CF}(\pi)\defeq\begin{cases}
            \textup{id}_1\star\textup{CF}(\pi|_1) & \text{if}\;\pi(1)=1\\
            \layer{0}{|\pi|-\pi^{-1}(1)}{\tau_{\pi^{-1}(1)-1}}\bullet(\textup{id}_1\star\textup{CF}(\pi|_1)) & \text{otherwise}
        \end{cases} 
    \end{gather*}
    where $\pi|_1\in\cat{Perm}(|\pi|-1)$ is the permutation satisfying $(\pi|_1)(i)\defeq\pi(i)-1$ when $1\le i<\pi^{-1}(1)$ and $(\pi|_1)(i)\defeq\pi(i+1)-1$ when $\pi^{-1}(1)\le i\le|\pi|-1$.
\end{definition}

Notice that $\cf{\Sigma_0}\subset\nf{\Sigma_0}\subset\pp{\Sigma_0}\subset\terms{\Sigma_0}$. Moreover, as stated in the following lemma, the canonizing map's definition is compatible with $\interp{\cdot}$.

\begin{lemma}\label{lem:cfcfterm}
    $\textup{CF}(\interp{t})=t$ for any $t\in \cf{\Sigma_0}$.
\end{lemma}
\begin{proof}
    In the following we use the notation $\pi^{\oplus n}\defeq \pi\oplus\dots\oplus\pi$ (with $n$ compositions $\oplus$). 
    If $t=\textup{id}_n$, then $\textup{CF}(\interp{\textup{id}_n})=\textup{CF}((1)^{\oplus n})=\textup{id}_n$. Otherwise, we proceed by induction on the number of layers. If $t=\layer{k}{\ell}{\tau_d}$, then $\textup{CF}(\interp{\layer{k}{\ell}{\tau_d}})=\textup{id}_k\star\textup{CF}(\interp{\layer{0}{\ell}{\tau_d}})=\textup{id}_k\star(\layer{0}{\ell}{\tau_d}\bullet(\textup{id}_1\star\textup{CF}((1)^{\oplus \ell+d})))=\layer{k}{\ell}{\tau_d}$. Moreover, if $t=\clayers{k\first ks}{\ell\first \ell s}{d\first ds}$, then we have the following proof.
    \begin{minipage}[b]{0.95\linewidth}
    \begin{gather*}
        \textup{CF}(\interp{C(k\first ks,\ell\first \ell s,d\first ds)})
        =\textup{id}_k\star\textup{CF}(\interp{C(0\first (ks\listminus k),\ell\first \ell s,d\first ds)}) \\
        =\textup{id}_k\star(\layer{0}{\ell+d-d}{\tau_d}\bullet(\textup{id}_1\star\textup{CF}(\interp{C(ks\listminus k\listminus 1,\ell s,ds)}))) \\[-0.4em]
        \overset{\textup{IH}}{=}\textup{id}_k\star(\layer{0}{\ell}{\tau_d}\bullet (\textup{id}_1\star C(ks\listminus k\listminus 1,\ell s,ds))) 
        =C(k\first ks,\ell\first \ell s,d\first ds)
    \end{gather*}
\end{minipage}%
\begin{minipage}[b]{0.05\linewidth}
    \qed
\end{minipage}
\end{proof}

\subsection{Terminating and Confluent Rewriting System}

We introduce the new rewrite rules in \cref{fig:permrs}. The rules \labelcref{ru:seqasso,ru:parasso,ru:seqidleft,ru:seqidright,ru:idid,ru:ididctx,ru:par,ru:parseqidleft,ru:parseqidright} (see \cref{fig:prors}) together with the rules \labelcref{ru:swap:swap0left,ru:swap:swap0right,ru:swap:swapreduce,ru:swap:layernat,ru:swap:layercon,ru:swap:layerfus,ru:swap:layercom,ru:swap:layernatctx,ru:swap:layerconctx,ru:swap:layerfusctx,ru:swap:layercomctx} form the rewriting system $\mathcal{R}_0$. Notice that $\mathcal{R}_0$ contains all rules of $\mathcal{R}$ except \labelcref{ru:layercom,ru:layercomctx}, which have been replaced by the rules \labelcref{ru:swap:layernat,ru:swap:layercon,ru:swap:layerfus,ru:swap:layercom,ru:swap:layernatctx,ru:swap:layerconctx,ru:swap:layerfusctx,ru:swap:layercomctx}. 

\begin{remark}
    \labelcref{ru:swap:layernat,,ru:swap:layercon,,ru:swap:layerfus,,ru:swap:layercom} explain how to transform a sequence of two layers $\layer{k_1}{\ell_1}{\tau_{d_1}}$ and $\layer{k_2}{\ell_2}{\tau_{d_2}}$ whenever $k_1\ge k_2$, and \labelcref{ru:swap:layernatctx,,ru:swap:layerconctx,,ru:swap:layerfusctx,,ru:swap:layercomctx} apply the same transformations when more layers follow. Here are some graphical examples.
    \begin{gather*}
        \tinytf{layernat_example_left}\rewrite{\eqref{ru:swap:layernat}}\tinytf{layernat_example_right} \hspace{3em}
        \tinytf{layercon_example_left}\rewrite{\eqref{ru:swap:layercon}}\tinytf{layercon_example_right} \\[0.2em]
        \tinytf{layerfus_example_left}\rewrite{\eqref{ru:swap:layerfus}}\tinytf{layerfus_example_right} \hspace{3em}
        \tinytf{layercom_example_left}\rewrite{\eqref{ru:swap:layercom}}\tinytf{layercom_example_right}
    \end{gather*}
\end{remark}

As stated in the following proposition, the rewriting system $\mathcal{R}_0$ is sound with respect to the coherence equations.
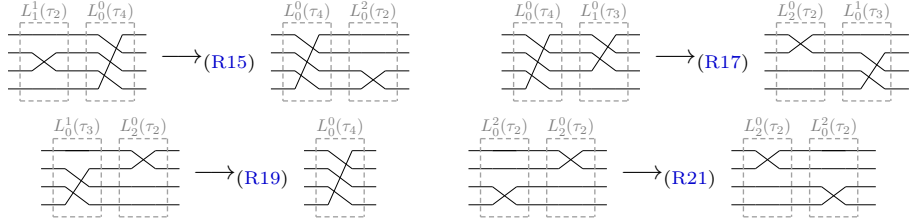
\begin{proposition}\label{prop:permsoundness}
    $t_1\mrewrite{\mathcal{R}_0}t_2\implies\cat{PRO}_{\Sigma_0,\mathcal{E}_0}\vdash t_1=t_2$ for any $t_1,t_2\in\pp{\Sigma_0}$.
\end{proposition}
\begin{proof}
    For any rewrite rule $(L\rewrite{}R)\in\mathcal{R}_0$ and any $K[\cdot]\in\ppcontexts{\Sigma_0}{\textup{dom}(L)}{\textup{cod}(L)}$, we prove $\cat{PRO}_{\Sigma_0,\mathcal{E}_0}\vdash K[L]=K[R]$ by induction on $K[\cdot]$ where the only non-trivial case is $K[\cdot]=[\cdot]$. For most rewrite rules this is trivial; the only interesting cases \labelcref{ru:swap:swapreduce,ru:swap:layernat,ru:swap:layercon,ru:swap:layerfus,ru:swap:layercom,ru:swap:layernatctx,ru:swap:layerconctx,ru:swap:layerfusctx,ru:swap:layercomctx} are proved using, in particular, \cref{eq:involution,eq:naturality,eq:symmetryinduct}. \qed
\end{proof}

\begin{figure}[t]
    \fbox{\begin{minipage}{0.982\linewidth}\begin{center}
        \vspace{-1em}
        \hspace{-0.8em}
        \begin{subfigure}{0.26\textwidth}
            \begin{equation}\label{ru:swap:swap0left}\tag{R12}
                \sigma_{0,n} \longrightarrow \textup{id}_n
            \end{equation}
        \end{subfigure}\hspace{3em}
        \begin{subfigure}{0.26\textwidth}
            \begin{equation}\label{ru:swap:swap0right}\tag{R13}
                \sigma_{n,0} \longrightarrow \textup{id}_n
            \end{equation}
        \end{subfigure}

        \vspace{-0.5em}
        \hspace{-0.8em}
        \begin{subfigure}{0.69\textwidth}
            \begin{equation}\label{ru:swap:swapreduce}\tag{R14}
                \sigma_{n+1,m+2} \longrightarrow (\sigma_{n+1,m+1}\otimes\textup{id}_1)\semicolon(\textup{id}_{m+1}\otimes\tau_{n+1})
            \end{equation}
        \end{subfigure}        

        \vspace{-0.5em}
        \hspace{-0.8em}
        \begin{subfigure}{0.66\textwidth}
            \begin{equation}\label{ru:swap:layernat}\tag{R15}
                \begin{array}{c}
                    \layer{k_1}{\ell_1}{\tau_{d_1}}\semicolon\layer{k_2}{\ell_2}{\tau_{d_2}} \longrightarrow \layer{k_2}{\ell_2}{\tau_{d_2}}\semicolon \layer{k_1+1}{\ell_1-1}{\tau_{d_1}} \\
                    \text{if}\;\; d_1\ge1,d_2\ge1 \;\;\text{and}\;\; k_2\le k_1<k_2+d_2-d_1
                \end{array}
            \end{equation}
        \end{subfigure}

        \vspace{-0.5em}
        \hspace{-0.8em}
        \begin{subfigure}{0.76\textwidth}
            \begin{equation}\label{ru:swap:layernatctx}\tag{R16}
                \begin{array}{c}
                    \layer{k_1}{\ell_1}{\tau_{d_1}}\semicolon (\layer{k_2}{\ell_2}{\tau_{d_2}}\semicolon t) \longrightarrow \layer{k_2}{\ell_2}{\tau_{d_2}}\semicolon (\layer{k_1+1}{\ell_1-1}{\tau_{d_1}}\semicolon t) \\
                    \text{if}\;\; d_1\ge1,d_2\ge1 \;\;\text{and}\;\; k_2\le k_1<k_2+d_2-d_1
                \end{array}
            \end{equation}
        \end{subfigure}

        \vspace{-0.5em}
        \hspace{-0.8em}
        \begin{subfigure}{0.84\textwidth}
            \begin{equation}\label{ru:swap:layercon}\tag{R17}
                \begin{array}{c}
                    \layer{k_1}{\ell_1}{\tau_{d_1}}\semicolon\layer{k_2}{\ell_2}{\tau_{d_2}} \longrightarrow \layer{k_2}{\ell_2+1}{\tau_{d_2-1}}\semicolon \layer{k_1+1}{\ell_1}{\tau_{d_1-1}} \\
                    \text{if}\;\; d_1\ge1,d_2\ge1 \;\;\text{and}\;\; k_2\le k_1\wedge k_2+d_2-d_1\le k_1<k_2+d_2
                \end{array}
            \end{equation}
        \end{subfigure}

        \vspace{-0.5em}
        \hspace{-0.8em}
        \begin{subfigure}{0.84\textwidth}
            \begin{equation}\label{ru:swap:layerconctx}\tag{R18}
                \begin{array}{c}
                    \layer{k_1}{\ell_1}{\tau_{d_1}}\semicolon (\layer{k_2}{\ell_2}{\tau_{d_2}}\semicolon t) \longrightarrow \layer{k_2}{\ell_2+1}{\tau_{d_2-1}}\semicolon (\layer{k_1+1}{\ell_1}{\tau_{d_1-1}}\semicolon t) \\
                    \text{if}\;\; d_1\ge1,d_2\ge1 \;\;\text{and}\;\; k_2\le k_1\wedge k_2+d_2-d_1\le k_1<k_2+d_2
                \end{array}
            \end{equation}
        \end{subfigure}

        \vspace{-0.5em}
        \hspace{-0.8em}
        \begin{subfigure}{0.55\textwidth}
            \begin{equation}\label{ru:swap:layerfus}\tag{R19}
                \begin{array}{c}
                    \layer{k_1}{\ell_1}{\tau_{d_1}}\semicolon\layer{k_2}{\ell_2}{\tau_{d_2}} \longrightarrow \layer{k_2}{\ell_1}{\tau_{d_1+d_2}} \\
                    \text{if}\;\; d_1\ge1,d_2\ge1 \;\;\text{and}\;\; k_1=k_2+d_2
                \end{array}
            \end{equation}
        \end{subfigure}

        \vspace{-0.5em}
        \hspace{-0.8em}
        \begin{subfigure}{0.63\textwidth}
            \begin{equation}\label{ru:swap:layerfusctx}\tag{R20}
                \begin{array}{c}
                    \layer{k_1}{\ell_1}{\tau_{d_1}}\semicolon (\layer{k_2}{\ell_2}{\tau_{d_2}}\semicolon t) \longrightarrow \layer{k_2}{\ell_1}{\tau_{d_1+d_2}}\semicolon t \\
                    \text{if}\;\; d_1\ge1,d_2\ge1 \;\;\text{and}\;\; k_1=k_2+d_2
                \end{array}
            \end{equation}
        \end{subfigure}

        \vspace{-0.5em}
        \hspace{-0.8em}
        \begin{subfigure}{0.63\textwidth}
            \begin{equation}\label{ru:swap:layercom}\tag{R21}
                \begin{array}{c}
                    \layer{k_1}{\ell_1}{\tau_{d_1}}\semicolon\layer{k_2}{\ell_2}{\tau_{d_2}} \longrightarrow \layer{k_2}{\ell_2}{\tau_{d_2}}\semicolon \layer{k_1}{\ell_1}{\tau_{d_1}} \\
                    \text{if}\;\; d_1\ge1,d_2\ge1 \;\;\text{and}\;\; k_1>k_2+d_2
                \end{array}
            \end{equation}
        \end{subfigure}

        \vspace{-0.5em}
        \hspace{-0.8em}
        \begin{subfigure}{0.74\textwidth}
            \begin{equation}\label{ru:swap:layercomctx}\tag{R22}
                \begin{array}{c}
                    \layer{k_1}{\ell_1}{\tau_{d_1}}\semicolon (\layer{k_2}{\ell_2}{\tau_{d_2}}\semicolon t) \longrightarrow \layer{k_1}{\ell_1}{\tau_{d_1}}\semicolon (\layer{k_2}{\ell_2}{\tau_{d_2}}\semicolon t) \\
                    \text{if}\;\; d_1\ge1,d_2\ge1 \;\;\text{and}\;\; k_1>k_2+d_2
                \end{array}
            \end{equation}
        \end{subfigure}
        \vspace{0.2em}
    \end{center}\end{minipage}}
    \caption{The rewrite rules are defined for any $n,m,k_1,k_2,\ell_1,\ell_2,d_1,d_2\in\N$ and \labelcref{ru:swap:layernat,ru:swap:layercon,ru:swap:layerfus,ru:swap:layercom,ru:swap:layernatctx,ru:swap:layerconctx,ru:swap:layerfusctx,ru:swap:layercomctx} are conditioned.}
    \label{fig:permrs}
\end{figure}

In order to prove termination and confluence, we build upon the result of \cref{sec:pro} and follow a similar proof flow.

\begin{proposition}\label{prop:permterminaison}
    $\mathcal{R}_0$ terminates.
\end{proposition}
\begin{proof}
    We build upon the proof of \cref{prop:proterminaison} by defining two additional positive quantities $\alpha,\gamma:\terms{\Sigma_0}\to\N$ and show that $\mathcal{R}_0$ strictly decreases $(\alpha,\beta,\gamma,\delta_2)$ with lexicographical order (notice that, as $\textup{def}(g)=0$ for any $g\in\Sigma_0$, we can take $D=2$ for $\delta_D$). The definitions of $\alpha$, $\beta$, $\gamma$ and $\delta_2$ are given in \cref{fig:quantities} together with a summary of which rewrite rules strictly decrease or preserve each quantity. Again, all the base cases have been verified with Isabelle/HOL. \qed
\end{proof}

\begin{lemma}\label{lem:cfispreserved}
    $t_1\rewrite{\mathcal{R}_0}^* t_2\implies\interp{t_1}=\interp{t_2}$ for any $t_1,t_2\in \pp{\Sigma_0}$.
\end{lemma}
\begin{proof}
    For any rewrite rule $(L\rewrite{}R)\in\mathcal{R}_0$ and any $K[\cdot]\in\ppcontexts{\Sigma_0}{\textup{dom}(L)}{\textup{cod}(L)}$, the equality $\interp{K[L]}=\interp{K[R]}$ is straightforwardly proved by induction on $K[\cdot]$ where the base cases $K[\cdot]=[\cdot]$ have been verified with Isabelle/HOL.\qed
\end{proof}

\begin{lemma}\label{lem:cfterminaison}
    $t\notrewrite{\mathcal{R}_0}\implies t\in \cf{\Sigma_0}$ for any $t\in \pp{\Sigma_0}$.
\end{lemma}
\begin{proof}
    All rewrite rules of $\mathcal{R}$ are in $\mathcal{R}_0$ except \labelcref{ru:layercom,ru:layercomctx}, which exactly correspond to \labelcref{ru:swap:layercom,ru:swap:layercomctx} in $\mathcal{R}_0$ when $g_i=\tau_{d_i}$ for some $d_i\ge1$. Moreover, every generator $\sigma_{n,m}$ in $t$ satisfies $n>0$ and $m=1$, otherwise a rule among \labelcref{ru:swap:swap0left,ru:swap:swap0right,ru:swap:swapreduce} applies. Thus, \cref{lem:nfterminaison} applies and we can assume that $t\in \nf{\Sigma_0}$. This implies that there exists $ks,\ell s,ds\in\N^p$ for some $p>0$ satisfying $ds_i\ge1$ for any $i\in\natset{p}$, such that $t=\clayers{ks}{\ell s}{ds}$. But then $\forall i\in\natset{p-1}:ks_i<ks_{i+1}$, otherwise a rule among \labelcref{ru:swap:layernat,ru:swap:layernatctx,ru:swap:layercon,ru:swap:layerconctx,ru:swap:layerfus,ru:swap:layerfusctx} applies. Hence, $t\in \cf{\Sigma_0}$. \qed
\end{proof}

\begin{proposition}\label{prop:permconfluence}
    $\mathcal{R}_0$ is confluent in $\pp{\Sigma_0}$.
\end{proposition}
\begin{proof}
    Let $t\in\pp{\Sigma_0}$. If $t\rewrite{\mathcal{R}_0}^*u$ and $t\rewrite{\mathcal{R}_0}^*v$, then by \cref{prop:permterminaison} there exist $u',v'\in\pp{\Sigma_0}$ such that $u\rewrite{\mathcal{R}_0}^*u'\notrewrite{\mathcal{R}_0}$ and $v\rewrite{\mathcal{R}_0}^*v'\notrewrite{\mathcal{R}_0}$. \cref{lem:cfcfterm,lem:cfispreserved,lem:cfterminaison} imply $\textup{CF}(\interp{t})=u'=v'$, which concludes the proof. \qed
\end{proof}

As stated in the following theorem, the termination and confluence of $\mathcal{R}_0$ solve our problem: given two terms $t_1,t_2\in\pp{\Sigma_0}$ expressing the same permutation, we can effectively rewrite $t_1$ into $t_2$ using the coherence equations.

\begin{theorem}\label{th:swap}
    $\interp{t_1}=\interp{t_2}\implies t_1\bothmrewrite{\mathcal{R}_0} t_2$ for any $t_1,t_2\in\pp{\Sigma_0}$.
\end{theorem}
\begin{proof}
    \cref{prop:permterminaison,lem:cfcfterm,lem:cfispreserved,lem:cfterminaison} imply that $t_i\rewrite{\mathcal{R}_0}^*\textup{CF}(\interp{t_i})$, and the assumption $\interp{t_1}=\interp{t_2}$ yields $t_1\bothmrewrite{\mathcal{R}_0} t_2$. \qed
\end{proof}


\section{Conclusion}\label{sec:conclusion}

In this work, we addressed the problem of diagrammatic equivalence for two classes of diagrams, namely $\cat{PRO}_\Sigma$ (for an arbitrary state-and-effect-free set of generators $\Sigma$) and $\cat{PRO}_{\Sigma_0,\mathcal{E}_0}=\cat{PROP}_\varnothing$, which can express any permutation. We used a term rewriting approach, providing terminating and confluent TRSs. Our plan is to combine our solutions to design a rewriting-based approach to handle diagrammatic equivalences in an arbitrary $\cat{PROP}_\Sigma$, the main bottleneck being that, in $\cat{PROP}_\Sigma$, generators can slide through swaps (see \cref{eq:naturality}), which breaks the normal forms we designed. Since our normal forms require the assumption of a state-and-effect-free PRO, it would also be interesting to generalize the first TRS to arbitrary PROs. Also, we could leverage the existing term-rewriting infrastructure present in Isabelle/HOL \cite{ceta} to obtain a verified checker for our TRSs, or create proof-checkers that compute the normal/canonical forms in a verified way following Sternagel and Thiemann \cite{standalonecertifier}. Our main objective---initiated by this first contribution---is to obtain a fully verified pipeline for certifying equivalences of quantum circuits.

\section*{Acknowledgements}
We thank Florian Frohn, Carsten Fuhs, René Thiemann, and Akihisa Yamada for insightful discussions on TRSs. We also thank the anonymous reviewers for their comments. This work is supported by the Plan France 2030 through the PEPR integrated project EPiQ ANR-22-PETQ-0007 and the HQI platform ANR-22-PNCQ-0002; and by the European project MSCA Staff Exchanges Qcomical HORIZON-MSCA-2023-SE-01. The project is also supported by the Maison du Quantique MaQuEst.

\bibliographystyle{splncs04}
\bibliography{ref}

\end{document}